\documentclass{llncs} 

\usepackage{amssymb}
\usepackage{textcomp}
\usepackage{amsmath}
\usepackage{algorithm}
\usepackage{algpseudocode}
\usepackage{comment}
\usepackage{graphicx}
\usepackage{color}
\newcommand{\Endproof}{\hfill$\Box$\\}

\begin{document}

\title{Quantum Algorithms for the Most Frequently String Search, Intersection of Two String Sequences and Sorting of Strings Problems}
\author{Kamil~Khadiev$^{1,2}$ ORCID: 0000-0002-5151-9908\and Artem Ilikaev$^{2}$} 

\institute{Smart Quantum Technologies Ltd., Kazan, Russia\and Kazan Federal University, Kazan, Russia \\ \email{kamil.hadiev@kpfu.ru, artemka.tema1998@gmail.com } 
}

\maketitle

\begin{abstract}
We study algorithms for solving three problems on strings. The first one is the Most Frequently String Search Problem. The problem is the following. Assume that we have a sequence of $n$ strings of length $k$. The problem is finding the string that occurs in the sequence most often. 
We propose a quantum algorithm that has a query complexity $\tilde{O}(n \sqrt{k})$. This algorithm shows speed-up comparing with the deterministic algorithm that requires $\Omega(nk)$ queries.

The second one is searching intersection of two sequences of strings. All strings have the same length $k$. The size of the first set is $n$ and the size of the second set is $m$.  We propose a quantum algorithm that has a query complexity $\tilde{O}((n+m) \sqrt{k})$. This algorithm shows speed-up comparing with the deterministic algorithm that requires $\Omega((n+m)k)$ queries.

The third problem is sorting of $n$ strings of length $k$. On the one hand, it is known that quantum algorithms cannot sort objects asymptotically faster than classical ones. On the other hand, we focus on sorting strings that are not arbitrary objects.  We propose a quantum algorithm that has a query complexity $O(n (\log n)^2 \sqrt{k})$. This algorithm shows speed-up comparing with the deterministic algorithm (radix sort) that requires $\Omega((n+d)k)$ queries, where $d$ is a size of the alphabet.

\textbf{Keywords:} quantum computation, quantum models, quantum algorithm, query model, string search, sorting \end{abstract}

\section{Introduction}
\emph{Quantum computing} \cite{nc2010,a2017} is one of the hot topics in computer science of last decades.
There are many problems where quantum algorithms outperform the best known classical algorithms \cite{dw2001,quantumzoo,ks2019,kks2019}.

One of these problems are problems for strings. Researchers show the power of quantum algorithms for such problems in  \cite{m2017,bbbv1997,rv2003}.

In this paper, we consider three problems:
\begin{itemize}
    \item the Most Frequently String Search problem;
    \item Strings sorting problem;
    \item Intersection of Two String Sequences problem.
\end{itemize}

 Our algorithms use some quantum algorithms as a subroutine, and the rest part is classical. We investigate the problems in terms of query complexity. The query model is one of the most popular in the case of quantum algorithms. Such algorithms can do a query to a black box that has access to the sequence of strings. As a running time of an algorithm, we mean a number of queries to the black box.

The first problem is the following. We have $n$ strings of length $k$. We can assume that symbols of strings are letters from any finite alphabet, for example,  binary, Latin alphabet or Unicode. The problem is finding the string that occurs in the sequence most often.
The problem \cite{ch2008} is one of the most well-studied ones in the area of data streams \cite{m2005,a2007datastreams,bcg2011,blm2015}. Many
applications in packet routing, telecommunication logging and tracking keyword queries in search machines are critically
based upon such routines.
The best known deterministic algorithms require $\Omega(nk)$ queries because an algorithm should at least test all symbols of all strings. The deterministic solution can use the Trie (prefix tree) \cite{d59,b98,b2008,knuth73} that allows to achieve the required complexity.

We propose a quantum algorithm that uses a self-balancing binary search tree for storing strings and a quantum algorithm for comparing strings. As a self-balancing binary search tree we can use the AVL tree \cite{avl62,cormen2001} or the Red-Black tree \cite{g78,cormen2001}. As a string comparing algorithm, we propose an algorithm that is based on the first one search problem algorithm from \cite{k2014,ll2015,ll2016}. This algorithm is a modification of Grover's search algorithm \cite{g96,bbht98}. Another important algorithm for search is described in \cite{lg2001}. Our algorithm for the most frequently string search problem has query complexity $O(n(\log n)^2 \cdot \sqrt{k})=\tilde{O}(n \sqrt{k})$, where $\tilde{O}$ does not consider a log factors. If $\log_2 n=o(k^{0.25})$, then our algorithm is better than deterministic one. Note, that this setup makes sense in practical cases.

The second problem is String Sorting problem. Assume that we have $n$ strings of length $k$. It is known \cite{hns2001,hns2002} that no quantum algorithm can sort arbitrary comparable objects faster than $O(n\log n)$. At the same time, several researchers tried to improve the hidden constant \cite{oeaa2013,oa2016}. Other researchers investigated space bounded case \cite{k2003sort}.  We focus on sorting strings. In a classical case, we can use an algorithm that is better than arbitrary comparable objects sorting algorithms. It is radix sort that has $O((n+d)k)$ query complexity \cite{cormen2001}, where $d$ is a size of the alphabet. Our quantum algorithm for the string sorting problem has query complexity $O(n(\log n)^2 \cdot \sqrt{k})=\tilde{O}(n \sqrt{k})$. It is based on standard sorting algorithms like Merge sort \cite{cormen2001} or Heapsort \cite{w1964,cormen2001} and the quantum algorithm for comparing strings. 

The third problem is the Intersection of Two String Sequences problem. Assume that we have two sequences of strings of length $k$. The size of the first set is $n$ and the size of the second one is $m$. The first sequence is given and the second one is given in online fashion, one by one. After each requested string from the second sequence, we want to check weather this string belongs to the first sequence.  We propose two quantum algorithms for the problem. Both algorithms has query complexity $O((n+m)\cdot\log n  \cdot\log(n+m) \sqrt{k})=\tilde{O}((n+m) \sqrt{k})$. The first algorithm uses a self-balancing binary search tree like the solution of the first problem. The second algorithm uses a quantum algorithm for sorting strings and has better big-$O$ hidden constant. At the same time, the best known deterministic algorithm requires $O((n+m)k)$ queries.

The structure of the paper is the following. We present the quantum subroutine that compares two strings in Section \ref{sec:compare}. Then we discussed three problems: the Most Frequently String Search problem in Section \ref{sec:freq}, Strings Sorting problem in Section \ref{sec:sort} and Intersection of Two String Sequences problem in Section \ref{sec:sets}. 

\section{The Quantum Algorithm for Two Strings Comparing}\label{sec:compare}
Firstly, we discuss a quantum subroutine that compares two strings of length $k$. Assume that this subroutine is $\textsc{Compare\_strings}(s,t,k)$ and it compares $s$ and $t$ in lexicographical order. It returns:
\begin{itemize}
    \item $-1$ if $s<t$;
    \item $0$ if $s=t$;
    \item $1$ if $s>t$;
\end{itemize}

As a base for our algorithm, we will use the algorithm of finding the minimal argument with $1$-result of a Boolean-value function. Formally, we have:
\begin{lemma}\cite{k2014,ll2015,ll2016}\label{lm:first-one}
Suppose, we have a function $f:\{1,\dots,N\}\to \{0,1\}$ for some integer $N$. There is a quantum algorithm for finding $j_0=\min\{j\in\{1,\dots,N\}:f(j)=1\}$. The algorithm finds $j_0$ with expected query complexity $O(\sqrt{j_0})$ and error probability that is at most $\frac{1}{2}$.
\end{lemma}

Let us choose the function $f(j)=(s_j\neq t_j)$. So, we search $j_0$ that is the index of the first unequal symbol of the strings. Then, we can claim that $s$ precedes $t$ in lexicographical order iff $s_{j_0}$ precedes $t_{j_0}$ in alphabet $\Sigma$. The claim is right by the definition of  lexicographical order. If there are no unequal symbols, then the strings are equal.

We use the standard technique of boosting success probability. So, we repeat the algorithm $3\log_2 n$ times and return the minimal answer, where $n$ is a number of strings in the sequence $s$. In that case, the error probability is $O\left(\frac{1}{2^{3\log n}}\right)=\left(\frac{1}{n^3}\right)$, because if we have an error in whole algorithm it means no invocation finds minimal index of unequal symbol.

Let us present the algorithm. We use $\textsc{The\_first\_one\_search}(f,k)$ as a subroutine from Lemma \ref{lm:first-one}, where $f(j)=(s_j\neq t_j)$. Assume that this subroutine returns $k+1$ if it does not find any solution.
\begin{algorithm}
\caption{$\textsc{Compare\_strings}(s,t,k)$. The Quantum Algorithm for Two Strings Comparing.}\label{alg:strcmp}
\begin{algorithmic}
\State $j_0 \gets  \textsc{The\_first\_one\_search}(f,k)$\Comment{The initial value}
\For{$i \in \{1,\dots,3\log_2 n\}$}
\State $j_0 \gets  \min(j_0, \textsc{The\_first\_one\_search}(f,k))$
\EndFor
\If{$j_0=k+1$} 
\State $result \gets 0$\Comment{The strings are equal.}
\EndIf
\If{$(j_0\neq k+1 )\& (s_{j_0}<t_{j_0})$}
\State $result \gets -1$ \Comment{$s$ precedes $t$.}
\EndIf
\If{$(j_0\neq k+1 )\& (s_{j_0}>t_{j_0})$}
\State $result \gets 1$ \Comment{$s$ succeeds $t$.}
\EndIf
\State \Return $result$
\end{algorithmic}
\end{algorithm}

The next property follows from the previous discussion.
\begin{lemma}\label{lm:strcmp}
Algorithm \ref{alg:strcmp} compares two strings of length $k$ in lexicographical order with  query complexity $O(\sqrt{k}\log n)$ and error probability $O\left(\frac{1}{n^3}\right)$.
\end{lemma}

%%%%%%%%%%%%%%%%%%%%%%%%%%%%%%%%%%%%%%%%%%%%%%%%
%%% The Most Frequently String Search Problem
%%%%%%%%%%%%%%%%%%%%%%%%%%%%%%%%%%%%%%%%%%%%%%%
\section{The Most Frequently String Search Problem}\label{sec:freq}
Let us formally present the problem.

{\bf Problem.}
For some positive integers $n$ and $k$, we have the sequence of strings $s=(s^1,\dots,s^n)$. Each $s^i=(s^i_1,\dots,s^i_k)\in \Sigma^k$ for some finite size alphabet $\Sigma$. Let $\#(s)=|\{i\in\{1,\dots,m\}: s^i=s\}|$ be a number of occurrences of string $s$. We search $s=argmax_{s^i\in S}\#(s^i)$.  

\subsection{The Quantum algorithm}

Firstly, we present an idea of the algorithm.

We use the well-known data structure a self-balancing binary search tree. As an implementation of the data structure, we can use the AVL tree \cite{avl62,cormen2001} or the Red-Black tree \cite{g78,cormen2001}. Both data structures allow as to find and add elements in $O(\log N)$ running time, where $N$ is a size of the tree.

The idea of the algorithm is the following. We store pairs $(i,c)$ in vertexes of the tree, where $i$ is an index of a string from $s$ and $c$ is a number of occurrences of the string $s^i$. We assume that a pair $(i,c)$ is less than a pair $(i',c')$ iff $s^i$ precedes $s^{i'}$ in the lexicographical order. So, we use $\textsc{Compare\_strings}(s^{i},s^{i'},k)$ subroutine as the compactor of the vertexes. The tree represents a set of unique strings from $(s^1,\dots,s^n)$ with a number of occurrences.

We consider all strings from $s^1$ to $s^n$ and check the existence of a string in our tree. If a string exists, then we increase the number of occurrences. If the string does not exist in the tree, then we add it. At the same time, we store $(i_{max},c_{max})=argmax_{(i,c)\mbox{ in the tree}}c$ and recalculate it in each step.

Let us present the algorithm formally. Let $BST$ be  a self-balancing binary search tree such that:
\begin{itemize}
    \item $\textsc{Find}(BST, s^i)$ finds vertex $(i,c)$ or returns $NULL$ if such vertex does not exist;
    \item $\textsc{Add}(BST, s^i)$ adds vertex $(i,0)$ to the tree and returns the vertex as a result;
     \item $\textsc{Init}(BST)$ initializes an empty tree;
\end{itemize}

\begin{algorithm}
\caption{The Quantum Algorithm for Most Frequently String Problem.}\label{alg:qmain}
\begin{algorithmic}
\State $\textsc{Init}(BST)$\Comment{The initialization of the tree.}
\State $c_{max}\gets 1$\Comment{The maximal number of occurrences.}
\State $i_{max}\gets 1$\Comment{The index of most requently string.}
\For{$i \in \{1,\dots,n\}$}
\State $v=(i,c) \gets \textsc{Find}(BST, s^i)$\Comment{Searching $s^i$ in the tree.}
\If{$v=NULL$}
\State $v=(i,c) \gets \textsc{Add}(BST, s^i)$\Comment{If there is no $s^i$, then we add it.}
\EndIf
\State $c\gets c+1$\Comment{Updating the vertex by increasing the number of occurrences.}
\If{$c>c_{max}$}\Comment{Updating the maximal value.}
\State $c_{max} \gets c$
\State $i_{max} \gets i$

\EndIf
\EndFor

\State \Return $s^{i_{max}}$
\end{algorithmic}
\end{algorithm}

Let us discuss the property of the algorithm.
\begin{theorem}\label{th:qfreq-compl}
Algorithm \ref{alg:qmain} finds the most frequently string from $s=(s^1,\dots,s^n)$ with query complexity $O(n(\log n)^2\cdot\sqrt{k})$ and error probability $O\left(\frac{1}{n}\right)$.
\end{theorem}
\begin{proof}
The correctness of the algorithm follows from the description.
Let us discuss the query complexity. Each operation $\textsc{Find}(BST, s^i)$ and $\textsc{Add}(BST, s^i)$ requires $O(\log n)$ comparing operations $\textsc{Compare\_strings}(s^{i},s^{i'},k)$. These operations are invoked $n$ times. Therefore, we have $O(n\log n)$ comparing operations. Due to Lemma \ref{lm:strcmp}, each comparing operation requires $O(\sqrt{k}\log n)$ queries. The total query complexity is  $O(n\sqrt{k}(\log n)^2)$.

Let us discuss the error probability. Events of error in the algorithm are independent. So, all events should be correct. Due to  Lemma \ref{lm:strcmp}, the probability of correctness of one event is $1-\left(1-\frac{1}{n^3}\right)$. Hence, the probability of correctness of all $O(n\log n)$ events is at least $ 1-\left(1-\frac{1}{n^3}\right)^{\alpha\cdot n\log n}$ for some constant $\alpha$.

Note that 
\[ \lim\limits_{n\to \infty} \frac{1-\left(1-\frac{1}{n^3}\right)^{\alpha\cdot n\log n}}{1/n}<1;\]
 Hence, the total error probability is at most $O\left(\frac{1}{n}\right)$.

\Endproof
\end{proof}

The data structure that we use can be considered as a separated data structure. We call it {\em ``Multi-set of strings with quantum comparator''}. Using this data structure, we can implement 
\begin{itemize}
    \item {\em ``Set of strings with quantum comparator''} if always $c=1$ in pair $(i,c)$ of a vertex;
    \item {\em ``Map with string key and quantum comparator''} if we replace $c$ by any data $r\in \Gamma$ for any set $\Gamma$. In that case, the data structure implements mapping $\Sigma^k\to\Gamma$.
\end{itemize}

All of these data structures has $O((\log n)^2 \sqrt{k})$  complexity of basic operations (\textsc{Find}, \textsc{Add}, \textsc{Delete}).

\subsection{On the Classical Complexity of the Problem}
The best known classical algorithm stores string to Trie (prefix tree) \cite{d59,b98}, \cite{b2008,knuth73} and do the similar operations. The running time of such algorithm is $O(nk)$. At the same time, we can show that if an algorithm tests$o(nk)$ variables, then it can return a wrong answer.

\begin{theorem}\label{th:dfreq-compl}
Any deterministic algorithm for the Most Frequently String Search problem has $\Omega(nk)$ query complexity.
\end{theorem}
\begin{proof}
Suppose, we have a deterministic algorithm $A$ for the Most Frequently String Search problem that uses $o(nk)$ queries. 

Let us consider an adversary that suggest an input. The adversary wants to construct an input such that the algorithm $A$ obtains a wrong answer.  

Without loss of generality, we can say that $n$ is even. Suppose, $a$ and $b$ are different symbols from an input alphabet. If the algorithm requests an variable $s^i_j$ for $i\leq n/2$, then the adversary returns $a$. If the algorithm requests an variable $s^i_j$ for $i>n/2$, then the adversary returns $b$. 

Because of the algorithm $A$ uses $o(nk)$ queries, there are at least one $s^{z'}_{j'}$ and one  $s^{z''}_{j''}$ that are not requested, where $z'\leq n/2$, $z''>n/2$ and $j',j''\in\{1,\dots,k\}$.

Let $s'$ be a string such that $s'_j=a$ for all $j\in\{1,\dots,k\}$. Let $s''$ be a string such that $s''_j=b$ for all $j\in\{1,\dots,k\}$.

Assume that $A$ returns $s'$. Then, the adversary assigns $s^{z'}_{j'}=b$ and assigns $s^i_j=b$ for each $i>n/2, j\in\{1,\dots,k\}$. Therefore, the right answer should be $s''$.

Assume that $A$ returns a string $s\neq s'$. Then, the adversary assigns $s^{z''}_{j''}=a$ and assigns $s^i_j=a$ for each $i\leq n/2, j\in\{1,\dots,k\}$. Therefore, the right answer should be $s'$.

So, the adversary can construct the input such that $A$ obtains a wrong answer.
\Endproof
\end{proof}
%%%%%%%%%%%%%%%%%%%%%%%%%%%%%%%%%%%%%%%%%%%%%%%%
%%% Strings Sorting Problem
%%%%%%%%%%%%%%%%%%%%%%%%%%%%%%%%%%%%%%%%%%%%%%%
\section{Strings Sorting Problem}\label{sec:sort}
Let us consider the following problem.

{\bf Problem.}
For some positive integers $n$ and $k$, we have the sequence of strings $s=(s^1,\dots,s^n)$. Each $s^i=(s^i_1,\dots,s^i_k)\in \Sigma^k$ for some finite size alphabet $\Sigma$. We search order $ORDER=(i_1,\dots, i_n)$ such that for any $j\in\{1,\dots,n-1\}$ we have $s^{i_j}\leq s^{i_{j+1}}$ in lexicographical order. 

We use Heap sort algorithm \cite{w1964,cormen2001} as a base and Quantum algorithm for comparing string from Section \ref{sec:compare}. We can replace Heap sort algorithm by any other sorting algorithm, for example, Merge sort \cite{cormen2001}. In a case of Merge sort, the big-O hidden constant in query complexity will be smaller. At the same time, we need more additional memory.

Let us present Heap sort for completeness of the explanation.
We can use Binary Heap \cite{w1964}. We store indexes of strings in vertexes. As in the previous section, if we compare vertexes $v$ and $v'$ with corresponding indexes $i$ and $i'$, then $v>v'$ iff $s^i>s^{i'}$ in lexicographical order. We use $\textsc{Compare\_strings}(s^{i},s^{i'},k)$ for comparing strings.
 Binary Heap $BH$ has three operations: 
 \begin{itemize}
    \item $\textsc{Get\_min\_and\_delete}(BH)$ returns minimal $s^i$ and removes it from the data structure.
    \item $\textsc{Add}(BH, s^i)$ adds vertex with value $i$ to the heap;
     \item $\textsc{Init}(BH)$ initializes an empty heap;
\end{itemize}
 
 The operations $\textsc{Get\_min\_and\_delete}$ and $\textsc{Add}$ invoke $\textsc{Compare\_strings}$ subroutine $\log_2 t$ times, where $t$ is the size of the heap. 
 
 The algorithm is the following.
 
 \begin{algorithm}\label{alg:sort}
\caption{The Quantum Algorithm for Sorting Problem.}
\begin{algorithmic}
\State $\textsc{Init}(BH)$\Comment{The initialization of the heap.}
\For{$i \in \{1,\dots,n\}$}
\State $\textsc{Add}(BH,s^i)$\Comment{Adding $s^i$ to the heap.}
\EndFor
\For{$i \in \{1,\dots,n\}$}
\State $ORDER\gets ORDER\cup\textsc{Get\_min\_and\_delete}(BH)$\Comment{Getting minimal string.}
\EndFor

\State \Return $ORDER$
\end{algorithmic}
\end{algorithm}
 
If we implement the sequence $s$ as an array, then we can store the heap in the same array. In this case, we do not need additional memory.

We have the following property of the algorithm that can be proven by the same way as Theorem \ref{th:qfreq-compl}.
\begin{theorem}\label{th:qsort-compl}
Algorithm \ref{alg:sort} sorts $s=(s^1,\dots,s^n)$ with query complexity $O(n(\log n)^2\cdot\sqrt{k})$ and error probability $O\left(\frac{1}{n}\right)$.
\end{theorem}

The lower bound for deterministic complexity can be proven by the same way as in Theorem \ref{th:dfreq-compl}.
\begin{theorem}\label{th:dsort-compl}
Any deterministic algorithm for Sorting problem has $\Omega(nk)$ query complexity.
\end{theorem}
The Radix sort \cite{cormen2001} algorithm almost reaches this bound and has $O((n+|\Sigma|)k)$ complexity. 
%%%%%%%%%%%%%%%%%%%%%%%%%%%%%%%%%%%%%%%%%%%%%%%%
%%% Intersection of Two Sequences of Strings Problem
%%%%%%%%%%%%%%%%%%%%%%%%%%%%%%%%%%%%%%%%%%%%%%%
\section{Intersection of Two Sequences of Strings Problem}\label{sec:sets}

Let us consider the following problem.

{\bf Problem.}
For some positive integers $n,m$ and $k$, we have the sequence of strings $s=(s^1,\dots,s^n)$. Each $s^i=(s^i_1,\dots,s^i_k)\in \Sigma^k$ for some finite size alphabet $\Sigma$. Then, we get $m$ requests $t=(t^1\dots t^m)$, where $t^i=(t^i_1,\dots,t^i_k)\in \Sigma^k$. The answer to a request $t^i$ is $1$ iff there is $j\in\{1,\dots,n\}$ such that $t^i=s^j$. We should answer $0$ or $1$ to each of $m$ requests.

We have two algorithms. The first one is based on {\em ``Set of strings with quantum comparator''} data structure from Section \ref{sec:freq}. We store all strings from $s$ to a self-balancing binary search tree $BST$. Then, we answer each request using $\textsc{Find}(BST, s^i)$ operation. Let us present the Algorithm \ref{alg:qintersection-set}.

\begin{algorithm}
\caption{The Quantum Algorithm for Intersection of Two Sequences of Strings Problem using {\em ``Set of strings with quantum comparator''} .}\label{alg:qintersection-set}
\begin{algorithmic}
\State $\textsc{Init}(BST)$\Comment{The initialization of the tree.}

\For{$i \in \{1,\dots,n\}$}
\State $\textsc{Add}(BST, s^i)$\Comment{We add $s^i$ to the set.}
\EndFor

\For{$i \in \{1,\dots,m\}$}
\State $v \gets \textsc{Find}(BST, t^i)$\Comment{We search $t^i$ in the set.}
\If{$v=NULL$} 
\State \Return $0$
\EndIf
\If{$v\neq NULL$} 
\State \Return $1$
\EndIf
\EndFor
\end{algorithmic}
\end{algorithm}

The second algorithm is based on Sorting algorithm from Section \ref{sec:sort}. We sort strings from $s$. Then, we answer to each request using binary search in the sorted sequence of strings \cite{cormen2001} and $\textsc{Compare\_strings}$ subroutine for comparing strings during the binary search. Let us present the Algorithm \ref{alg:qintersection-binsearch}. Assume that the sorting Algorithm \ref{alg:sort} is the subroutine $\textsc{Sort\_strings}(s)$ and it returns the order $ORDER=(i_1,\dots,i_n)$. The binary search algorithm with $\textsc{Compare\_strings}$ subroutine as comparator is $\textsc{Binary\_search\_for\_strings}(t,s, OREDER)$ subroutine and it searches $t$ in the ordered sequence $(s^{i_1},\dots,s^{i_n})$. Suppose that the subroutine $\textsc{Binary\_search\_for\_strings}$ returns $1$ if it finds $t$ and $0$ otherwise.

\begin{algorithm}
\caption{The Quantum Algorithm for Intersection of Two Sequences of Strings Problem using sorting algorithm .}\label{alg:qintersection-binsearch}
\begin{algorithmic}
\State $ORDER\gets \textsc{Sort\_strings}(s)$\Comment{We sort $s=(s^1,\dots,s^n)$.}

\For{$i \in \{1,\dots,m\}$}
\State $ans \gets \textsc{Binary\_search\_for\_strings}(t,s, OREDER)$\Comment{We search $t^i$ in the ordered sequence.}
\State \Return $ans$
\EndFor
\end{algorithmic}
\end{algorithm}

The algorithms have the following query complexity.

\begin{theorem}
Algorithm \ref{alg:qintersection-set} and Algorithm \ref{alg:qintersection-binsearch} solve Intersection of Two Sequences of Strings Problem with query complexity $O((n+m)\sqrt{k}\cdot\log n\cdot\log(n+m))$ and error probability $O\left(\frac{1}{n+m}\right)$.
\end{theorem}
\begin{proof}
The correctness of the algorithms follows from the description.
Let us discuss the query complexity of the first algorithm. As in the proof of Theorem \ref{th:qfreq-compl}, we can show that constructing of the search tree requires  $O(n\log n)$ comparing operations. Then, the searching of all strings $t^i$ requires $O(m\log n)$ comparing operations. The total number of comparing operations is $O((m+n)\log n)$. We will use little bit modified version of the Algorithm \ref{alg:strcmp} where we run it $3(\log (n + m))$ times. We can prove that comparing operation requires $O(\sqrt{k}\log (n+m))$ queries. The proof is similar to the proof of corresponding claim from the proof of Lemma \ref{lm:strcmp}. So, the total complexity is $O((n+m)\sqrt{k}\cdot\log n\cdot\log(n+m))$.

The second algorithm also has the same complexity because it uses  $O(n\log n)$ comparing operations for sorting and  $O(m\log n)$ comparing operations for all invocations of the binary search algorithm.

Let us discuss the error probability. Events of error in the algorithm are independent. So, all events should be correct. We can prove that the error probability for comparing operation is $O(1/(n+m)^3)$. The proof is like the proof of Lemma \ref{lm:strcmp}. So, the probability of correctness of one event is $1-\left(1-\frac{1}{(n+m)^3}\right)$. Hence, the probability of correctness of all $O((n+m)\log n)$ events is at least $ 1-\left(1-\frac{1}{(n+m)^3}\right)^{\alpha\cdot(n+m)\log n}$ for some constant $\alpha$.

Note that 
\[ \lim\limits_{n\to \infty} \frac{1-\left(1-\frac{1}{(n+m)^3}\right)^{\alpha\cdot (n+m)\log n}}{1/(n+m)}<1;\]
 Hence, the total error probability is at most $O\left(\frac{1}{n+m}\right)$.

\Endproof
\end{proof}

Note that Algorithm \ref{alg:qintersection-binsearch} has a better big-$O$ hidden constant than Algorithm \ref{alg:qintersection-set}, because the Red-Black tree or AVL tree has a height that greats $\log_2 n$ constant times. So, adding elements to the tree and checking existence has bigger big-$O$ hidden constant than sorting and binary search algorithms.

The lower bound for deterministic complexity can be proven by the same way as in Theorem \ref{th:dfreq-compl}.
\begin{theorem}
Any deterministic algorithm for Intersection of Two Sequences of Strings Problem has $\Omega((n+m)k)$ query complexity.
\end{theorem}

This complexity can be reached if we implement the set of strings $s$ using Trie (prefix tree) \cite{d59,b98,b2008,knuth73}. 

Note, that we can use the quantum algorithm for element distinctness \cite{a2007elementDist}, \cite{a2004} for this problem. The algorithm solves a problem of finding two identical elements in the sequence. The query complexity of the algorithm is $O(D^{2/3})$, where $D$ is a number of elements in the sequence. The complexity is tight because of \cite{as2004}. The algorithm can be the following. On $j$-th request, we can add the string $t^j$ to the sequence $s^1,\dots,s^n$ and invoke the element distinctness algorithm that finds a collision of $t^j$ with other strings. Such approach requires $\Omega(n^{2/3})$ query for each request and $\Omega(mn^{2/3})$ for processing all requests. Note, that the streaming nature of requests does not allow us to access to all $t^1,\dots,t^m$ by Oracle. So, each request should be processed separately.
%%%%%%%%%%%%%%%%%%%%%%%%%%%%%%%%%%%%%%%%%%%%%%%%
%%% Conclusion
%%%%%%%%%%%%%%%%%%%%%%%%%%%%%%%%%%%%%%%%%%%%%%%
\section{Conclusion}\label{sec:concl}
In the paper we propose a quantum algorithm for comparing strings. Using this algorithm we discussed four data structures:   {\em ``Multi-set of strings with quantum comparator''}, {\em ``Set of strings with quantum comparator''}, {\em ``Map with a string key and quantum comparator''} and {\em ``Binary Heap of strings with quantum comparator''}. We show that the first two data structures work faster than the implementation of similar data structures using Trie (prefix tree) in a case of $\log_2 n = o(k^{0.25})$. The trie implementation is the best known classical implementation in terms of complexity of simple operations (add, delete or find). Additionally, we constructed a quantum strings sort algorithm that works faster than the radix sort algorithm that is the best known deterministic algorithm for sorting a sequence of strings.

Using these two groups of results, we propose quantum algorithms for two problems: the Most Frequently String Search and Intersection of Two String Sets.  These quantum algorithms are more efficient than deterministic ones.

\subsection*{Acknowledgement}
This work was supported by Russian Science Foundation Grant 19-71-00149.
We thank Aliya Khadieva, Farid Ablayev and Kazan Federal University quantum group for useful discussions.
\bibliographystyle{alpha}
\bibliography{tcs}
\newpage

\end{document}